\documentclass[pra,twocolumn,showpacs,nobibnotes]{revtex4-2} 
\usepackage{amsmath}
\usepackage{amssymb}
\usepackage{enumerate}
\usepackage{graphicx,helvet}
\usepackage{color}	
\usepackage{bm}                      
\usepackage{tikz}
 \usepackage{listings}
 \usepackage{soul}
\usepackage{amsfonts}
 \usepackage{amsmath}
\usepackage{lipsum}
\definecolor{mygray}{gray}{0.4}
\definecolor{light-blue}{rgb}{0.8,0.85,1}
\graphicspath{{images/}}
\usepackage{amsmath}%
\usepackage{amsthm}
\usepackage{MnSymbol}%
\usepackage{wasysym}%
\usepackage{bbm}
\usetikzlibrary{decorations.shapes}
\usepackage[draft,inline,nomargin]{fixme} \fxsetup{theme=color} 
\FXRegisterAuthor{mz}{mzz}{\color{magenta}MZ}
\FXRegisterAuthor{cp}{acp}{\color{blue}CP}
\newcommand{\eref}[1]{eq.~(\ref{#1})} 
\FXRegisterAuthor{dd}{ddd}{\color{red} DD}
\FXRegisterAuthor{sf}{ssf}{\color{green} SF}
\usepackage{hyperref}
\usepackage[capitalize]{cleveref}
\hypersetup{
  linkcolor=red,
  colorlinks=true
}

\mathchardef\Re="023C
\mathchardef\Im="023D
\newcommand{\prlsection}[1]{\textit{#1.}---}

\newcommand{\mcI}{\mathcal{I}}
\newcommand{\mcB}{\mathcal{B}}
\newcommand{\mcH}{{\sf H}}
\newcommand{\mcM}{\mathcal{M}}
\newcommand{\mcN}{\mathcal{N}}

\newcommand{\mcY}{\mathcal{Y}}
\newcommand{\mcX}{\mathcal{X}}
\newcommand{\mcZ}{\mathcal{Z}}
\newcommand{\mcE}{\mathcal{E}}

\usepackage{bbm}

\newcommand{\mcL}{\mathcal{L}}
\newcommand{\Jami}{Choi-Jamio\l{}kowski}
\newcommand{\mcF}{\mathcal{F}}
\newcommand{\mcT}{\mathcal{T}}

\newcommand{\mcD}{\mathcal{D}}
\newcommand{\diag}{\text{diag}}
\newcommand{\ie}{i.e.}

\newcommand{\rmi}{\mathrm{i}}

\newcommand{\fref}[1]{fig.~\ref{#1}}

\newcommand{\tr}{\mathop{\mathrm{tr}}\nolimits}

\newcommand{\id}{\mcI}

\newcommand{\ket}[1]{{\vert #1 \rangle}}
\newcommand{\bra}[1]{{\langle #1 \vert}}
\newcommand{\proj}[2]{{\vert #1 \rangle \langle #2 \vert}}
\graphicspath{{images/}}

\newtheorem{theorem}{Theorem}
\newtheorem{proposition}{Proposition}

\newtheorem*{theorem6}{Theorem 6}
\newtheorem*{theorem3}{Theorem 3}
\newtheorem*{theorem1}{Theorem 1}
\newtheorem*{proposition1}{Proposition 1}
\newcommand{\Ind}{\ensuremath{\overline{{\sf C}^{\rm div}}}}

\newcommand{\cptp}{{\ensuremath{\sf C}}}
\newcommand{\cptpd}[1]{\ensuremath{{\sf C}_#1}}

\newcommand{\cpDiv}{\ensuremath{{\sf C}^\text{CP}}}
\newcommand{\LDiv}{\ensuremath{{\sf C}^\text{L}}}
\newcommand{\one}{\openone}

\newcommand{\unam}{Universidad Nacional Aut\'onoma de M\'exico, Ciudad de M\'exico 01000, Mexico}
\newcommand{\ifunam}{Instituto de F\'{\i}sica, \unam}

\begin{document}
\title{Quantum dynamics is not strictly bidivisible \\[2ex]       
{\normalfont \normalsize  \textcolor{black}{``Dedicado a la memoria de Juan Manuel Dávalos Ramírez''}} }
\author{David Davalos} \email[]{davidphysdavalos@gmail.com}\affiliation{\sas}\affiliation{\ifunam}
\author{Mario Ziman}  \email[]{mario.ziman@savba.sk}\affiliation{\sas}
\newcommand{\sas}{Institute of Physics, Slovak Academy of Sciences, D\'ubravsk\'a cesta 9, Bratislava 84511, Slovakia}
\email{davidphysdavalos@gmail.com}
\begin{abstract} 
We address the question of the existence of quantum channels that are divisible in two quantum channels but not in three or, more generally, channels divisible in $n$ but not in $n+1$ parts. We show that for the qubit those channels \textit{do not} exist, whereas for general finite-dimensional quantum channels the same holds at least for full Kraus rank channels. To prove these results, we introduce a novel decomposition of quantum channels which separates them into a boundary and Markovian part, and it holds for any finite dimension. Additionally, the introduced decomposition amounts to the well-known connection between divisibility classes and implementation types of quantum dynamical maps, and can be used to implement quantum channels using smaller quantum registers.
%
\end{abstract} 
\pacs{03.65.Yz, 03.65.Ta, 05.45.Mt}
\maketitle
\prlsection{Introduction}
Quantum channels are one of the basic building blocks of quantum physics describing fixed time transformations of quantum systems~\cite{cirac,zimansbook,rivasreview}. They are used to describe memoryless noise in quantum communication~\cite{wildebook} or the decoherence processes during quantum computation. The famous Stinespring dilation theorem~\cite{Stinespring2006} guarantees that any quantum channel $\mcE$ can be represented by a system-environment Hamiltonian $H$, a fixed time $t_\text{fixed}$ and a suitable initial state of the environment $\varrho_\text{E}$. In particular, the Schr\"odinger equation defines the following transformation of the quantum system
\begin{equation}
\mcE[\varrho]=\tr_\text{E} \left[ U(t_\text{fixed}) \left(\varrho \otimes \varrho_\text{E} \right) U^\dagger(t_\text{fixed})\right]\,,
\label{eq:open_system}
\end{equation}
where $U(t)=e^{-\rmi t H}$ (taking $\hbar=1$) and $\tr_\text{E}$ is the partial trace over the environment. Such an open system model of quantum processes suggests that the induced quantum channel can be understood as a composition of shorter (both in time and induced changes) state transformations. However, as was discovered in the seminal work by Wolf and Cirac~\cite{cirac}, there exist quantum channels that cannot be written as a concatenation of other channels thus; they are indivisible. This is similar to the prime numbers; they cannot be factorized. In this Letter, we investigate this analogy in more detail and show its powerful applications to structural problems of quantum channels.

We are interested to see how a given channel can be factorized
into indivisible ones. In particular, our aim is to characterize
the families of $n$-divisible quantum channels, i.e., the channels that
are concatenations of at most $n$ quantum channels.
As we will see, there are several key differences between
divisibility and factorization. First, the special role of unity
is played by the class of unitary channels and the appearance
of unitary channels in the decomposition is considered as trivial and does
not count as being an indivisible channel. Second, the concatenation
is not unique; thus, there are different ways the channel can be expressed
as composition of nontrivial channels. Third, there exist infinitely
divisible channels that can be expressed as a concatenations of infinitely
many quantum channels. Their existence follows from the solution of the Gorini-Kossakowski-Sudarshan-Lindblad equation~\cite{lindblad,Gorini1976,kossa,kossa2}, also known as the time-dependent Markovian master equation (see~\eref{eq:canonicalform}) describing the Markovian time evolution of open quantum systems. For example, in the case of time-independent master equation, the solution reads in the form $\mcE_t=e^{\mcL t}$, where $\mcL$ is the time-independent Lindblad generator. If a channel can be expressed in this form, we call it Markovian. The exponential form
implies $\mcE_t$ can be expressed as the $n$th power of quantum channel
$\mcF_n=e^{\mcL t/n}$, i.e., $\mcE_t=(\mcF_n)^n$ for arbitrary integer $n$. 

A prominent example of indivisible qubit channel is the so-called optimal quantum NOT process transforming a state $\varrho$ into a noisy version of its ``orthogonal'' state $\varrho^\perp=\one-\varrho$. In particular,
$\mcE_{\rm NOT}(\varrho)=\frac{1}{3}(\one+\varrho^\perp)$.
Applying this indivisible channel twice, we obtain
\begin{eqnarray}
\nonumber
\mcE_{\rm NOT}^2(\varrho)
&=&\frac{2}{3}\mcE_{\rm NOT}\left(\frac{1}{2}\one\right)+\frac{1}{3}\mcE_{\rm NOT}\left(\varrho^\perp\right)\\\nonumber
&=&\frac{2}{9}\left(\one+\frac{1}{2}\one\right)+\frac{1}{9}\left(\one+\varrho\right)=
\frac{8}{9}\frac{1}{2}\one+\frac{1}{9}\varrho\,,
\end{eqnarray}
thus, $\mcE_{\rm NOT}^2=\frac{1}{9}\mcI+\frac{8}{9}\mcN$, where
$\mcI$ is the identity map (noiseless quantum channel)
and $\mcN$ is the quantum channel transforming all states into the
total mixture state (also known as the completely depolarizing channel).
Since $\mcN^2=\mcN$, $\mcI^2=\mcI$, and $\mcI\circ\mcN=\mcN\circ\mcI=\mcN$,
it follows that the $n$th power of $\mcD_q=q\mcI+(1-q)\mcN$ equals
$\mcD_q^n=q^n\mcI+(1-q^n)\mcN=\mcD_{q^n}$. Setting $q=\sqrt[n]{1/9}$,
we obtain $\mcD_{\sqrt[n]{1/9}}^n=\mcE_{\rm NOT}^2$; hence it is $n$ divisible.
In summary, this example illustrates that concatenating two indivisible
maps might result in an infinitely divisible one. We see that not only is the division not
unique, but it might also be qualitatively different. 

In this Letter, we will show that the set of at most $n$-divisible (with $n\geq 2$)  qubit
channels is empty for the qubit case, because all of them are either divisible in infinite parts
or indivisible. To do this, we will introduce a specific decomposition of general (finite-dimensional) quantum channels into a boundary element of the set of channels
and a Markovian channel, i.e., $\mcE=e^\mcL \mcE_\text{boundary}$, where $\mcL$ is
a suitable Lindbladian. Interestingly, for the qubit case we will additionally show that $\mcE=e^\mcL \mcE_\text{indivisible}$ holds for any non-Markovian channel.

\par

%

\prlsection{Quantum channels and divisibility}
Physics of quantum systems is most commonly formulated within the framework of
the associated $d$-dimensional complex Hilbert spaces $\mcH_d$. We will
assume the dimension is finite. The set of linear operators $\mcB(\mcH_d)$
on $\mcH_d$ contains density operators representing quantum states, and
quantum channels $\mcE$ are \emph{completely positive trace-preserving linear
maps} acting on $\mcB(\mcH_d)$, i.e., $\tr{\mcE(X)}=\tr{X}$ for all
$X\in\mcB(\mcH_d)$ and $(\id_k \otimes \mcE)(A)\geq 0$ for all
positive operators $A\in\mcB(\mcH_k\otimes\mcH_d)$ and all 
integers $k\geq 1$ determining the dimension of the Hilbert space $\mcH_k$.
Let us denote by $\cptpd{d}$ the set of all quantum channels.
Every channel can be (nonuniquely) expressed in the so-called
operator-sum form as follows: $\mcE[X]=\sum_j K_j X K_j^{\dagger}$,
where $\sum_j K_j^\dagger K_j=\one$. The minimum number of operators $K_j$
required in the previous expression is called the \emph{Kraus rank}
of $\mcE$. If $\mcE(\one)=\one$, the channel is \emph{unital}.
If $\mcE[X]=UXU^\dagger$ for some \textit{unitary operator}
$U$ (meaning $UU^\dagger=U^\dagger U=\one$), we say the channel is
\textit{unitary}.

The set of quantum channels $\cptpd{d}$ is convex and
closed under the composition. A quantum channel $\mcE$ is called
\emph{indivisible} if it cannot be written as a concatenation of
two nonunitary channels, namely, if $\mcE =
\mcE_1\mcE_2$ implies that either $\mcE_1$ or $\mcE_2$, exclusively, is a
unitary channel. If the channel is not indivisible, it is said to be
\emph{divisible}.  We denote the set of divisible channels by ${\sf C}^{\rm
  div}$ and that of indivisible channels by \Ind{}.
Following this definition, unitary channels are divisible, because for them
both (decomposing) channels $\mcE_{1,2}$ must be unitary. 
The concept of indivisible channels resembles the concept of prime numbers:
Unitary channels play the role of unity (which are not indivisible or prime), i.e.,
a composition of indivisible and a unitary channel results in an indivisible
channel.

The ability to divide quantum channels into smaller ones is intimately related with the concept of continuous time evolution, especially with the question of how a given channel can emerge from time evolution. This question has been explored since the seminal work of Evans~\cite{Evans1977}, where it was discovered that the subset of quantum channels $\LDiv$ achievable by (time-independent) Lindblad master equations is quite limited. In particular, 
$\mcE\in\LDiv{}$ if $\mcE=e^\mcL$, where the Lindblad generator $\mcL$
is defined as follows:
\begin{equation}
\mcL[\varrho]= i [\varrho,H]
 +\sum_{\alpha,\beta=0}^{d^2-1}  G_{\alpha \beta}
     \left( 
         F_{\alpha}\varrho F^{\dagger}_{\beta}
             -\frac{1}{2} \lbrace F^{\dagger}_{\beta} F_{\alpha},\varrho \rbrace 
     \right),
\label{eq:canonicalform}
\end{equation}
with $G\geq 0$ and $\left\{F_\alpha \right\}_{\alpha=0}^{d^2-1}$ form an orthonormal basis of $\mcB(\mcH_d)$, i.e., $\tr{F_\alpha^\dagger F_\beta}=\delta_{\alpha\beta}$ with $\tr F_i:=\delta_{i0}/\sqrt{d}$.
Recently, several classes of divisibility have been
found~\cite{cirac,davalosdivisibility}. Let us introduce the set of
$n$-divisible channels $\cptp^{(n)}$. We say $\mcE\in\cptp^{(n)}$ if there
exist a collection of channels $\mcE_1,\dots,\mcE_n$ such that
$\mcE=\mcE_1\circ\cdots\circ\mcE_n$. Clearly, $\cptp^{(n+1)}\subset\cptp^{(n)}$
and our goal is to characterize quantum channels inside
$\cptp^{(n)}\setminus\cptp^{(n+1)}$.

Let us recall that \textit{full} Kraus rank channels are divisible \cite{cirac}, and, in the case of unital qubit channels, only Kraus rank three channels are indivisible. Using the representation of qubit unital channels from Ref.~\cite{ruskai}, in which the set of unital channels is represented by a tetrahedron (see~\fref{fig:1}), the indivisible channels correspond to the faces of the tetrahedron (excluding the edges); thus, they form a subset of measure zero.

In what follows, we are ready to state the main theorem.


\par
\begin{theorem}[Lindblad-Boundary decomposition]
\label{thm:main}
Any $\mcE \in \cptpd{d}$ can be written as follows:
\begin{equation}
\mcE=e^\mcL \mcE_\text{boundary},
\label{eq:main_formula}
\end{equation} 
where $\mcE_\text{boundary}$ is a channel in the boundary between $\cptp_d$ and trace-preserving maps, \ie{} it has Kraus rank less than $d^2$~\cite{zimansbook}, with $d=\dim \mcH$, and $\mcL$ is a Lindbladian.
\end{theorem}
Given a channel $\mcE$, the logic of the proof relies in finding families with the form $\mcF_t=e^{-t \mcL} \mcE$ such that for large $t$, they no longer parametrize quantum channels. This, together with continuity arguments, let us prove that there exists a value of $t$ such that $\mcF_t$ is a boundary channel. Then, using the invertibility of $e^\mcL$ we arrive to the desired decomposition. The formal proof is given in appendix~\ref{app:thm1}; there, we stress that the singular case needs a special treatment. Some results from the literature are used for the proof~\cite{cirac,choi,chuangbook}.

\par
Let us note that for nonsingular channels the proof is independent of the order; thus, also the decomposition $\mcE= \mcE_\text{boundary} e^\mcL$ is possible. However, this is no longer the case for singular channels. Consider, for instance, the completely depolarizing channel $\mcN[\Delta]:=\one/d \tr \Delta$ and assume that $\mcN= \mcE_\text{boundary} e^\mcL$. Since $e^\mcL$ is invertible, we can solve for $\mcE_\text{boundary}$, but $\mcN e^{-\mcL}=\mcN=\mcE_\text{boundary}$ for all $\mcL$; this is a contradiction, and, therefore, the decomposition $\mcN= \mcE_\text{boundary} e^\mcL$ is not possible.
An example of the Lindblad-boundary decomposition is given for the depolarizing channel in the appendix~\ref{app:example}.
\par
As a direct consequence we obtain the following theorem.
\par
\begin{theorem}[Nonexistence of strict $n$-divisibility ($n\geq 2$) for full Kraus rank channels]
\label{thm:full_rank}
Every finite-dimensional quantum channel, $\mcE \in \cptpd{d}$, with full Kraus rank, is divisible in an infinite (also uncountable) number of channels.
\end{theorem}
%



\par
Let us now introduce a result concerning the freedom of $\mcL$.

\par
\begin{proposition}[Pure dissipative choice]
For any channel, the Lindblad-boundary (LB) decomposition can always be performed by choosing the Lindbladian $\mcL$ to be  purely dissipative, \ie{} $\mcL[\Delta]=\sum_{\alpha, \beta=0}^{d^2-1} G_{\alpha \beta}\left(F_\alpha \Delta F_\beta^\dagger{}- \frac{1}{2} \left\{ F_\beta^\dagger{} F_\alpha,\Delta \right \} \right)$; see \eref{eq:canonicalform}.
\end{proposition} 
\begin{proof} For non-singular channels, only $G$ participates in the proof of Theorem \ref{thm:main}; thus, we can always simply omit the Hamiltonian part of any given generator. On the other hand, for the nonsingular channels a suitable pure dissipative generator must be found (see the appendix~\ref{app:prop1}).
\end{proof}
\par
\prlsection{Divisibility of qubit channels}
Before discussing the results that can be derived using the novel decomposition stated by Theorem \ref{thm:main}, let us first discuss corrections of two theorems needed to characterize the boundary of $\cptp{}_2$ from the point of view of the divisibility. The known results rely on Lorentz normal forms~\cite{Verstraete2001}, and they arise as an analogous decomposition to singular value decomposition but taking the Lorentz metric instead of the Euclidean.
\par

Furthermore, we will exploit Theorem 8 in Ref.~\cite{Verstraete2002} (being an adaptation of Theorem 3 in Ref.~\cite{Verstraete2001}). It turns out this theorem needs to be extended (see Theorem~\ref{theorem2} below) and corrected, because its proof does not cover all Lorentz normal forms. An explicit counterexample has been reported in the appendix in  Ref.~\cite{davalosdivisibility}.

%
\begin{figure}
\centering
\includegraphics[width=5cm]{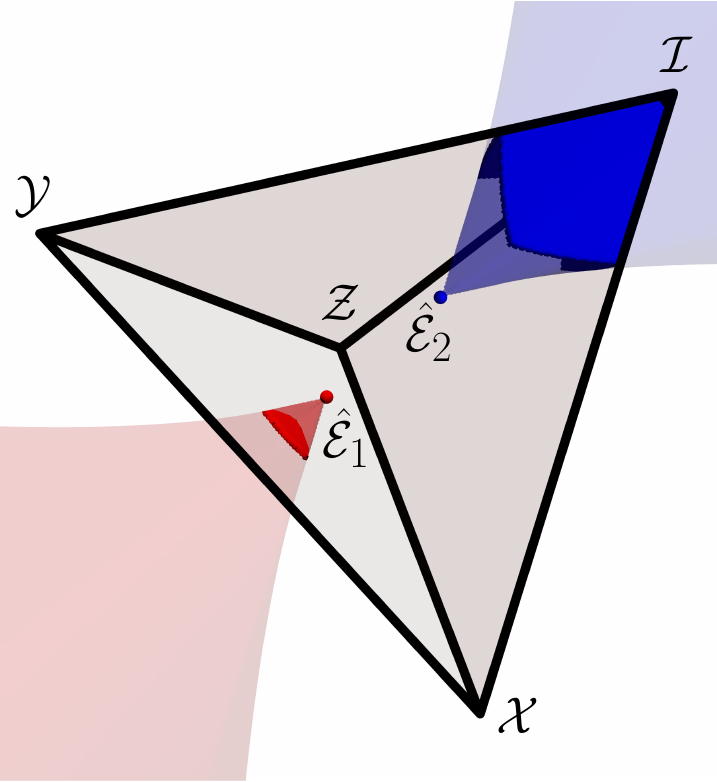}
\caption{Regions of trace-preserving Pauli diagonal maps corresponding to two families with the form $e^{-\mcL t}\mcE$, $t \geq 0$ (with $e^{\mcL t}$ Pauli diagonal with positive eigenvalues~\cite{davalosdivisibility}). The family corresponding to $\hat\mcE_1=\text{diag}(1,-1/5,-1/5,-1/5)$ (red dot), which is not infinitesimally divisible, crosses only indivisible channels (Kraus rank three); this is a manifestation of Theorem \ref{thm:lb_qubit}. On the other hand, the family corresponding to $\hat\mcE_2=\text{diag}(1,1/5,1/5,1/5)$ (blue dot) crosses both divisible and indivisible channels. $\mcX$, $\mcY$, and $\mcZ$ denote the Pauli unitaries, solid colors are for the crossings, and transparent and light transparent are for the interior and exterior of the tetrahedron, respectively.
\label{fig:1}
}
\end{figure}

\par
\begin{theorem}[Lorentz normal forms for qubit channels]
\label{theorem2}
Given a qubit quantum channel, there exist two Kraus rank one (not necessarily nonincreasing trace) linear maps over $ \mcB(\mcH)$, $\mcT_1$ and $\mcT_2$ such that $\mcE = \mcT_2 \mcM \mcT_1$, where $\mcM$ is one of the following forms (in the Pauli basis):
\begin{enumerate}
\item $\hat \mcM$ is diagonal---\ie{}, $ \mcM$ is a Pauli channel---or
\item $\hat \mcM$ is nondiagonal and has the following form,
 \begin{equation}
\hat \mcM(v,x,z):=\left(
\begin{array}{cccc}
 1 & 0 & 0 & z \\
 0 & x f(v,z) & 0 & 0 \\
 0 & 0 & x f(v,z) & 0 \\
 v & 0 & 0 & v-z+1 \\
\end{array}
\right),
\end{equation}
\end{enumerate}
with $f(v,z)=\sqrt{1+v-z-vz}$. 
The form $\mcM(v,x,z)$ has at most Kraus rank 3 and $\mcM(v,1,z)$ at most Kraus rank 2.
\end{theorem}
\par
Using Theorem~\ref{theorem2} (which proof can be found in the appendix), we can patch the proof of Theorem 23 in Ref.~\cite{cirac}, which characterizes indivisible qubit channels.
%
This theorem states that a qubit channel is indivisible if and only if it has diagonal Lorentz normal form with Kraus rank 3. The problematic part is the discarding of channels having nondiagonal Lorentz normal form as indivisible. The original proof relies on noticing that the nondiagonal case can always be written as a concatenation of two channels with at most Kraus rank 2. Fortunately a similar trick holds also for the missing cases that we announced in Theorem~\ref{theorem2}; observe that
\begin{equation}
\hat \mcM(v,x,z)=\hat \mcM(v,1,z)\diag\left(1,x,x,1 \right)
\end{equation}
where $\hat \mcM(v,1,z)$ and $\diag\left(1,x,x,1 \right)$ have at most Kraus rank two. Now we can define the CP maps $\hat \mcE_\mcM=\hat \mcT_2 \hat\mcM(v,1,z) \hat \mcT_1$ and $\hat \mcE_D= \hat\mcT_2 \diag(1,x,x,1) \hat \mcT_1$, such that $\hat \mcE=\hat\mcE_\mcM \hat \mcT \hat\mcE_D$ with $\mcT=\mcT_1^{-1}\mcT_2^{-1}$ (a CP Kraus rank one operation), given that $\mcT_{1,2}$ are invertible. Then we have a concatenation of two not necessarily trace-preserving CP maps with at most Kraus rank 2, $\mcE_\mcM$ and $\mcT \mcE_D$. To finish, Theorem 12 in Ref.~\cite{cirac} guarantees that $\mcE$ can be written as a concatenation of two trace-preserving maps with at most Kraus rank 2. Since Kraus rank 2 maps are infinitesimally divisible (${\sf C}^{\text{CP}}$) (see Theorem 19 in Ref.~\cite{cirac}), this completes the correction to the proof.

\par
Having 
Theorem 23 in Ref.~\cite{cirac} valid, we can safely establish that the boundary of the set of qubit channels is completely characterized in terms of the divisibility types discussed here. And what is particularly relevant is that there are only indivisible channels (the ones having Kraus rank 3 and diagonal Lorentz normal form) and infinitesimally divisible channels (the ones having Kraus rank 3 with nondiagonal Lorentz form, Kraus rank 2 and trivially Kraus rank 1). Using these facts, we can prove the following.

\par
\begin{theorem}[Lindblad-Boundary decomposition of non-Markovian qubit channels]
\label{thm:lb_qubit}
Let $\mcE \in \cptp{}_2$ be a qubit channel that is not infinitesimally divisible, \ie{}, $\mcE \not\in \cpDiv{}$, and then its Lindblad-boundary decomposition reads
\begin{equation}
\mcE=e^\mcL \mcE_\text{indivisible}\,.
\label{eq:markov_non_markov}
\end{equation}
Moreover, if $\mcE$ is indivisible, then $e^\mcL$ is unitary.
\end{theorem}
\par
\begin{proof}
The proof is straightforward if we recall the construction of one-parametric families of the form $\mcF_t=e^{-t\mcL} \mcE$. We have already proved there always exists some $t_\text{min}$ such that $\mcF_{t_\text{min}}$ is a boundary channel. As mentioned above, we may observe it can be only infinitesimally divisible or indivisible. If $\mcF_{t_\text{min}}$ is infinitesimally divisible, then $\mcE$ must be also infinitesimally divisible; thus, we have a contradiction. Therefore, we conclude that $\mcF_{t_\text{min}}$ is indivisible. Thus, qubit channels are either indivisible, infinitesimally divisible, or with the form in~\eref{eq:markov_non_markov}, \ie{} divisible in an arbitrary number of parts due to the infinitesimal divisibility of $e^\mcL$ (see~\fref{fig:1}). This motivates the following theorem.
\end{proof}

\par
\begin{theorem}[Nonexistence of strictly $n$-divisible ($n\geq 2$) qubit channels]
\label{thm:qubits_non_existence}
Every divisible qubit channel is divisible in an infinite (in fact, uncountable) number of channels. Therefore, $\cptpd{2}^{(n)}\setminus \cptpd{2}^{(n+1)}$ is empty for all integer $n \geq 2$ or, equivalently, $\cptpd{2}^{(n)}=\cptpd{2}^{\rm div}$ for all $n\geq 2$.
\end{theorem}
\par
\prlsection{Reduction of the ancilla size and its limitations}
Notice that the Lindblad-boundary decomposition can be used to reduce the size of the ancilla needed to simulate or implement experimentally arbitrary quantum channels. Since the boundary part is always Kraus rank deficient, the trick is to find a Kraus rank deficient Lindblad part, too.
For the qubit case, this size reduction of the ancilla can be enough even to dispense with a qubit. This is stated with the following theorem, which is proved in the appendix~\ref{app:thm6}.
\begin{theorem}[Divisibility in Kraus rank deficient channels]
Let $\mcE \in \cptp{}_2$ be a qubit channel---it is divisible in channels with at most Kraus rank $2$ if and only if $\mcE$ is infinitesimally divisible ($\mcE\in\cpDiv$); otherwise, at least one factoring channel has Kraus rank no less than $3$.
\end{theorem}
Therefore, for infinitesimally divisible qubit channels we can use a smaller quantum register to implement it in a quantum computer (two qubits instead of three, one for the system and one for the ancilla). We constructed such circuit and computed its quantum process tomography in an IBM \textit{falcon r4T} quantum processor~\cite{qiskit,python}; the averaged fidelity obtained is almost equal to $1$, and ten trials are shown in~\fref{fig:circuit}. See the appendixes~\ref{app:qc1} and~\ref{app:qc2} for further details.
\begin{figure}
\centering
\includegraphics[width=0.98\columnwidth]{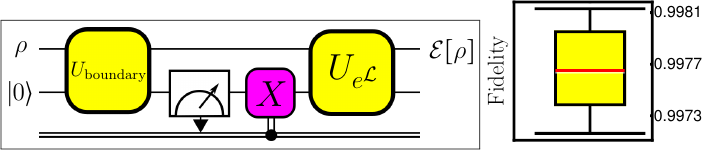}
\caption{Circuit with a two-qubit register and a bit to implement infinitesimally divisible qubit channels ($\cpDiv{}$) in a quantum computer (left). Box plot of fidelities between the computed Choi matrix (in an IBM \textit{falcon r4T} quantum processor) and the theoretical one of the channel $\diag(1,0,0,3/4)$ (in the Pauli basis) (right). The red line indicates the average. The quantum process tomography was done $10$ times; in each of them, there were performed the $12$ independent experiments needed, $20\times 10^3$ times each. See the appendixes~\ref{app:qc1} and~\ref{app:qc2} for further details.\label{fig:circuit}}
\end{figure}
\par
\prlsection{Summary and outlook}
We proved that strictly $n$-divisible full Kraus rank channels do not exist for $n \geq 2$ (Theorem~\ref{thm:full_rank}). For the case of qubits, this nonexistence applies to all channels (Theorem~\ref{thm:qubits_non_existence}), thus, divisibility always implies divisibility in an arbitrary number of channels. For the general case, the question remains open for the channels from the boundary between completely positive and noncompletely positive trace-preserving maps. This suggests that the analogy between integer factorization and channel divisibility, that motivates this investigation, leads to significant differences. Moreover, the concept of constructing channels from ``prime'' channels is qualitatively different from the case of integers. Most likely, we will meet with the same situation if we consider the case of classical channels. However, it would be interesting to see whether there are some differences between classical and quantum channels from the perspective of the divisibility structures.

In order to obtain these results, we introduced and investigated a novel Lindblad-boundary decomposition that resembles polar decomposition but including a dissipation term (Theorem~\ref{thm:main}). To ensure the validity of the results for the qubit case, we patched a theorem from the literature and fixed the proof of the characterization of indivisible channels (Theorem~\ref{theorem2}).
The novel decomposition allow us to construct methods to reduce the size of the ancilla needed to simulate quantum channels in a quantum computer. Moreover, we believe that these results constitute a useful tool for the analysis of quantum channels and motivates further foundational studies of the structural questions on the dynamics of open quantum systems.
\begin{acknowledgements}
\prlsection{Acknowledgments} Support by the projects OPTIQUTE APVV-18-0518, CONACyT 285754, DESCOM VEGA-2/0183/21, and Štefan Schwarz Support Fund is acknowledged, as well conversations with Carlos Pineda, Thomas Gorin, Sergey Filippov, and Tomás Basile. 
\end{acknowledgements}
\clearpage
\begin{figure}[h!]
\centering
\includegraphics[scale=0.2]{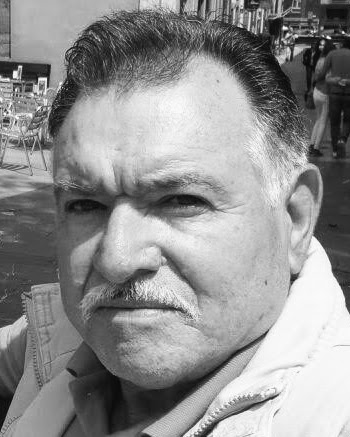}
\end{figure}
\parbox[c]{7cm}{
\textit{\begin{center}
(From David Davalos) I dedicate this letter to the memory of my Father; not only in the sense of remembrance, but also in the sense of his own memories and beautiful mind, which exist somewhere.
\end{center}}}
\vspace{0.7cm}
\bibliographystyle{apsrev4-1}
\bibliography{labibliografia}
\appendix
\section{Proof of Theorem 1}
\label{app:thm1}
\begin{theorem1}[Lindblad-boundary decomposition]
Any $\mcE \in \cptpd{d}$ can be written as follows,
\begin{equation}
\mcE=e^\mcL \mcE_\text{boundary},
\label{eq:main_formula}
\end{equation} 
where $\mcE_\text{boundary}$ is a channel in the boundary between $\cptp_d$ and trace-preserving maps, \ie{} it has Kraus rank less than $d^2$~\cite{zimansbook}, with $d=\dim \mcH$, and $\mcL$ is a Lindbladian.
\end{theorem1}
\begin{proof}
For non-singular channels the proof is straightforward using the properties of the determinant, in particular we stress that if $\mcE \in \cptpd{d}$ then $|\det \mcE |\leq 1$~\cite{cirac}. 
Thus, let $\mcE \in \cptp_d$ with $|\det \mcE| >0 $ and define the family $\mcF_t=e^{-\mcL t} \mcE$ with $t\in \mathbb{R}^+$ where $\mcL$ is \textit{any} Lindblad operator. Now let us compute the determinant $\det \mcF_t=\det e^{-t \mcL} \det \mcE=e^{t d\tr G} \det \mcE$. It is clear that $|\det \mcF_t|>1$ for $t> -(d \tr G)^{-1} \log \left(|\det \mcE |\right)=:t'$, therefore $\mcF_t$ is no longer inside $\cptpd{d}$ for $t> t'$. Notice that this does not imply that $t'$ is the transition time outside $\cptpd{d}$ as the condition $|\det \mcF_t|>1$ is only sufficient. The exact transition time can be computed for each map using \Jami{} isomorphism~\cite{choi}.
%
Since $\mcF_t$ parameterizes a continuous curve inside trace-preserving maps with $\mcF_0=\mcE$ due to the continuity of $e^{-t \mcL}$, then $\mcF_t$ crosses the boundary of $\cptpd{d}$ once, otherwise $e^{t \mcL}$ could not be a valid channel for all $t \geq 0$, which is a contradiction. Therefore there exists some $t_\text{min}>0$ such that $\mcE_\text{boundary}:=\mcF_{t_\text{min}}$. 
Since $e^{-t_\text{min} \mcL}$ is invertible given that $t_\text{min}$ is finite, and  $t_\text{min}\mcL$ is a Lindblad operator, we proved the desired relation for invertible channels.
For singular channels we cannot longer take advantage of the properties of determinant. To prove this case, a suitable Lindblad generator must be found. Consider the following continuous family of quantum channels,
\begin{equation}
\mcE_{\psi,t}[\Delta]:=e^{-t \mu} \Delta +(1-e^{- t \mu}) \proj{\psi}{\psi} \tr \Delta 
\label{eq:psi_map}
\end{equation}
 for all $\Delta \in \mcB(\mcH)$ with $\mu>0$ and $t\geq 0$. It is easy to prove that $\mcE_{\psi,t_1}\mcE_{\psi, t_2}=\mcE_{\psi, t_1+t_2}$, \ie{} $\mcE_{\psi,t}=e^{t \mcL_\psi}$~\cite{lindblad} with $\mcL_\psi[\Delta]= \mu \left( \proj{\psi}{\psi} \tr \Delta -\Delta \right)$ the Lindblad generator. We then define the family $\mcF_{\psi,t} := e^{- t \mcL_\psi} \mcE$ for a given channel $\mcE$ (not necessarily singular). Notice that $\mcF_{\psi,t}\neq \mcE$ unless $\mcE[\Delta]=\proj{\psi}{\psi}\tr\Delta$, thus, we can always choose $\ket{\psi}$ such that the family $\mcF_{\psi,t}$ is non-trivial. Consider now a density matrix $\sigma$ such that $\varrho=\mcE[\sigma]\neq \proj{\psi}{\psi}$, and define $\varrho':=\mcF_{\psi,t}[\sigma]=e^{t \mu} \varrho +(1-e^{t \mu}) \proj{\psi}{\psi}$, by construction $\varrho$ is a density matrix. Now we compute $\bra{\psi}\varrho' \ket{\psi}=e^{t \mu} \left( \bra{\psi}\varrho \ket{\psi}-1 \right)+1$ and observe that $\bra{\psi}\varrho'\ket{\psi} <0$ for $t > -\mu^{-1}\log\left(1- \bra{\psi}\varrho \ket{\psi}\right)=:t'>0$, provided that $0\leq\bra{\psi}\varrho \ket{\psi}<1$. Therefore $\mcF_{\psi,t}$ is not even positive for $t>t'$. Following  similar arguments from the non-singular case, we can find some $t_\text{min}$ such that $\mcE=e^{\mcL}\mcE_\text{boundary}$ with $\mcL:=t_\text{min}\mcL_\psi$ and $\mcE_\text{boundary}:=\mcF_{\psi,t_\text{min}}$.
\end{proof}

\section{Example: Lindblad-boundary decomposition for completely depolarizing channel}
\label{app:example}
Following the same steps used in the proof of Theorem 1, we define the continuous family $\mcF_{\psi, t}=e^{-t \mcL_\psi}\mcN$. The task now is to determine the time $t_\text{min}$ when the curve crosses the boundary of $\cptpd{d}$, to do so we use the \Jami{} representation of $\mcF_{\psi, t}$, it is given by $\tau_\mcF=\left( \id \otimes e^{-t \mcL_\psi}\mcN \right)[\omega]=e^{t \mu} \one_{d^2}/d^2+(1-e^{\mu t})\left( \one_d/d\right)\otimes \proj{\psi}{\psi}$, where $\omega$ is the projector of a Bell state between the system and a copy of it. $\mcF$ is completely positive if and only if $\tau_\mcF \geq 0$~\cite{choi}, since it is hermitian, $t_\text{min}$ coincides with the time for the smaller eigenvalue to be negative. Computing the spectral decomposition of $\tau_\mcF$ is straightforward if one writes $\one_{d^2}$ using an orthogonal separable basis $\left\{\ket{\phi_i}\otimes \ket{\psi_j} \right\}_{i,j=1}^{d}$ with $\ket{\psi_1}:=\ket{\psi}$. Thus, the eigenvalues of $d\tau_\mcF$ are $e^{\mu t} \left( 1/d-1 \right)$ and $e^{\mu t}/d$, with multiplicities $d$ and $d^2-d$, respectively. Therefore $t_\text{min}=\mu^{-1}\log\left(d/(d-1)\right)$ and $\mcE_\text{boundary}$ has Kraus rank $d^2-d$. Let us illustrate this result for the qubit using the Pauli basis, $1/\sqrt{2}\left\{ \sigma_0,\sigma_x,\sigma_y,\sigma_z\right\}$ with $\sigma_0:=\one_2$. Thus, the matrix components of a channel $\mcE$ are $\hat \mcE_{i,j}=\frac{1}{2}\tr{\sigma_i\mcE[\sigma_j}]$. In this basis $\hat \mcN =\diag \left( 1,0,0,0 \right)$, and the decomposition is
\begin{equation}
\hat \mcN=\underbrace{\left(
\begin{array}{cccc}
 1 & 0 & 0 & 0 \\
 0 & \frac{1}{2} & 0 & 0 \\
 0 & 0 & \frac{1}{2} & 0 \\
 \frac{1}{2} & 0 & 0 & \frac{1}{2} \\
\end{array}
\right)}_{e^{t_{\min} \hat \mcL_{\psi}}}
\underbrace{\left(
\begin{array}{cccc}
 1 & 0 & 0 & 0 \\
 0 & 0 & 0 & 0 \\
 0 & 0 & 0 & 0 \\
 -1 & 0 & 0 & 0 \\
\end{array}
\right)}_{\hat \mcE_\text{boundary}}
\end{equation}
where we have chosen $\ket{\psi}=\ket{0}$ with $\sigma_z \ket{0}=\ket{0}$. Observe that $\mcE_\text{boundary}=\lim_{t\to \infty}\mcE_{\psi,t}$ (taking $\ket{\psi}=\ket{1}$ with $\sigma_z\ket{1}=-\ket{1}$), therefore $\mcE_\text{boundary}$ lies in the boundary of \LDiv{} for this case.

\section{Proof of Proposition 1}
\label{app:prop1}
\begin{proposition1}[Pure dissipative choice]
For any channel the Lindblad-Boundary (LB) decomposition can always be performed by choosing the Lindbladian $\mcL$ to be  purely dissipative, \ie{} $\mcL[\Delta]=\sum_{\alpha, \beta=0}^{d^2-1} G_{\alpha \beta}\left(F_\alpha \Delta F_\beta^\dagger{}- \frac{1}{2} \left\{ F_\beta^\dagger{} F_\alpha,\Delta \right \} \right)$.
\end{proposition1} 
\begin{proof}
For the non-singular channels only $G$ participates in the proof of Theorem 1, thus, we can always simply omit the Hamiltonian part of any given generator. On the other hand, for the non-singular channels a suitable pure dissipative generator must be found.
Let $d$ be the dimension of the system's Hilbert space and $\left \{ \ket{i} \right\}_{i=0}^{d-1}$ an orthonormal basis, and define the family of pure dissipative Lindblad generators that model spontaneous decay of energy levels described by states ranging from $\ket{1}$ to $\ket{d-1}$, to the level $\ket{0}$, with decaying ratios $\gamma_1,\dots \gamma_{d-1}$. Thus,  
\begin{equation}
\mcL_{\text{AD}}[\Delta]=\sum_{i=1}^{d-1} \gamma_i \left(F_i \Delta F_i^\dagger{}-\frac{1}{2} \left\{ F_i^\dagger{} F_i,\Delta \right\}\right),
\label{eq:amplitude_damping}
\end{equation}
with $F_i=\proj{0}{i}$, \ie{} $\tr F_i=0$, so the Hamiltonian part is null. Now, similar to the proof of Theorem 1 we define the one-parametric family $\mcF_t^\text{AD}=e^{-t \mcL_\text{AD}} \mcE$. Now let $\varrho$ be a density matrix such that $\varrho=\mcE[\sigma]$ with $\sigma$ some density matrix, assume  the spectral decomposition $\varrho=\sum_{k=0}^{d-1} \lambda_k\proj{k}{k}$. We then choose the basis that defines $F_i$ as the eigenbasis of $\varrho$, \ie{} $F_i\varrho=\lambda_i \proj{0}{i}$ and $\proj{0}{0}\neq \varrho$ (in case of $\varrho$ being pure) and compute the following,
\begin{equation*}
\bra{0} e^{-t \mcL_\text{AD}} \left[\proj{k}{k}\right] \ket{0} =
\begin{cases}
1-e^{\gamma_k t} & \text{for }k>0\\
1 & \text{for } k=0.
\end{cases}
\end{equation*}
Therefore $\bra{0} \mcF_t^\text{AD}[\sigma]\ket{0}=\bra{0}e^{-t \mcL_\text{AD}}[\varrho]\ket{0}=\lambda_0+\sum_{k=1}^{d-1} \lambda_k \left( 1-e^{\gamma_k t} \right)$, it is clear that there exist a large enough $t'$ such that $\bra{0} \mcF_{t'}^\text{AD}[\sigma]\ket{0}<0$, \ie{} $\mcF_{t'}^\text{AD}$ is not even positive. Therefore, using similar arguments as for the non-singular case of Theorem 1, there exist some $t_\text{min}$ such that $\mcE_\text{boundary}:=\mcF_{t_\text{min}}^\text{AD}$ and $\mcE=e^{\mcL}\mcE_\text{boundary}$ with $\mcL:=t_\text{min} \mcL_\text{AD}$ pure dissipative.
\end{proof}

\section{Proof of theorem 3}
\label{app:thm3}
\begin{theorem3}[Lorentz normal forms for qubit channels]
Given a qubit quantum channel, there exist two Kraus rank one (not necessarily non-increasing trace) linear maps over $ \mcB(\mcH)$, $\mcT_1$ and $\mcT_2$ such that $\mcE = \mcT_2 \mcM \mcT_1$, where $\mcM$ is one of the following forms (in the Pauli basis):
\begin{enumerate}
\item $\hat \mcM$ is diagonal, \ie{} $ \mcM$ is a Pauli channel.
\item $\hat \mcM$ is non-diagonal and has the following form,
 \begin{equation}
\hat \mcM(v,x,z):=\left(
\begin{array}{cccc}
 1 & 0 & 0 & z \\
 0 & x f(v,z) & 0 & 0 \\
 0 & 0 & x f(v,z) & 0 \\
 v & 0 & 0 & v-z+1 \\
\end{array}
\right),
\end{equation}
\end{enumerate}
with $f(v,z)=\sqrt{1+v-z-vz}$. 
For the latter, the following cases can be identified,
\begin{enumerate}[i.)]
\item $z\in [0,1)$, $v\in (-1,z]$, $x\in (-1,1)$ (Kraus rank 3 for $v<z$ and Kraus rank 2 for $v=z$).
\item $z\in [0,1)$, $v\in (-1,z]$, $x\in \{-1,1\}$ (Kraus rank 2 for $v<z$ and Kraus rank 1 for $v=z$).
\item $z\in [0,1)$, $v=-1$ and $x=0$ (Kraus rank 2 for $z>-1$ and Kraus rank 1 for $z=-1$).
\item $z=1$, $x=0$, $v\in [-1,1]$ (Kraus rank 2 for $|v|<1$ and Kraus rank 1 for $|v|=1$).
\end{enumerate}
\end{theorem3}

\begin{proof}
Let $\tau_\mcE=\frac{1}{4}\sum_{i,j=0}^3 R_{ij}\sigma_i \otimes \sigma_j$ be the \Jami{} state of $\mcE$, and let $\hat \mcE$ be the matrix of $\mcE$ in the Pauli basis and $R$ the matrix formed with coefficients $R_{ij}$, then the identity $ R=\hat\mcE \Phi_\text{T} $, with $\Phi_\text{T}=\diag{\left(1,1,-1,1 \right)}$, holds~\cite{davalosdivisibility}. According to Theorem 3 in Ref.~\cite{Verstraete2001}, we can write $R=L_2 \Sigma L^\text{T}_1$ where $L_{1,2}$ are proper orthochronous Lorentz transformations, corresponding to stochastic local operations~\footnote{Quantum operations are defined to be completely positive non-increasing trace linear maps over $\mcB(\mcH)$, \ie{} $\mcE$ is CP with $\tr\mcE[X]\leq \tr X \ \forall X\in \mcB(\mcH)$} and classical communication at the level of $\tau_\mcE$, and $\Sigma$ has one of the following forms,
\begin{equation}
\Sigma_1 =\diag(s_0,s_1,s_2,s_3), \ \ 
\Sigma_2 =\left(\begin{array}{cccc}
a & 0 & 0 & b \\ 
0 & d & 0 & 0 \\ 
0 & 0 & -d & 0 \\ 
c & 0 & 0 & a+c-b
\end{array}  \right),
\end{equation}
with $s_0\geq s_1 \geq s_2 \geq |s_3|$. For the case of $\Sigma_1$ the corresponding quantum operation~\cite{Note1} is $ \hat \mcM=\Sigma_1 \Phi_\text{T}/s_0$~\cite{davalosdivisibility} (\ie{} Pauli) and $\hat \mcT_2=s_0 L_2$; for $\Sigma_2$ we have $\hat \mcM=\Sigma_2\Phi_\text{T}/a$ and $\hat \mcT_2=a L_2$. For both cases we have $\hat \mcT_1=\Phi_\text{T} L^\text{T}_1 \Phi_\text{T}$. Defining $z:=b/a$, $v:=c/a$, $d'=d/a$ and constructing the \Jami{} matrix of $\mcM$, $\tau_\mcM$, a direct evaluation of $\tau_\mcM \geq 0$ (\ie{} the condition for $\mcM$ to be completely positive (CP)) can be done by computing its eigenvalues. This procedure is simple but tedious, as it relies in the analysis of four inequalities, here we sketch only the analysis: Choosing $d'=0$, cases iii.) and iv.) become evident; cases i.) and ii.) arise by observing that $\tau_\mcM \geq 0$ implies $|d'|\leq f(v,z)$, then we can write $d'=x f(v,z)$ with $x\in [-1,1]$. The analysis of the Kraus rank is straightforward by observing the number of non-zero eigenvalues of $\tau_\mcM$ for each case, this finishes the proof.
\end{proof}

\section{Proof of theorem 6}
\label{app:thm6}
\begin{theorem6}[Divisibility in Kraus rank deficient channels]
A qubit channel, $\mcE \in \cptp{}_2$, is divisible in channels with at most Kraus rank $2$ if and only if it is infinitesimally divisible ($\mcE\in\cpDiv$), otherwise, at least one factoring channel has Kraus rank no less than $3$.
\end{theorem6}
\begin{proof}
To prove the theorem we use the fact that a qubit channel is infinitesimally divisible if and only if its Lorentz normal form is infinitesimally divisible too (theorem 17 in Ref.~\cite{cirac}).
Now, to prove the first part, assume that $\mcE$ is an infinitesimally divisible qubit channel. If it has Kraus rank $3$, it always has non-diagonal Lorentz normal form (otherwise it would be indivisible due to theorem 23 in Ref.~\cite{cirac}), non-diagonal forms are always divisible in Kraus rank $2$ channels according to theorem 3 (and theorem 19 in Ref.~\cite{cirac}). 
If it has Kraus rank $4$ and is non-singular, according to theorem 4 in Ref.~\cite{davalosdivisibility}, its Lorentz normal form is a Pauli channel with the form $e^{\hat\mcL}=\diag\left(1,\eta_1,\eta_2,\eta_3\right)$ (in the Pauli basis with $\eta_i>0$, and up to unitary conjugations). Therefore the components $\eta_i$ fulfill the following~\cite{davalosdivisibility},
\begin{equation}
\frac{\eta_i}{\eta_j \eta_k}\geq 1, \ \  i\neq j\neq k.
\end{equation}
Then, we can find three positive numbers $\lambda_i\leq 1 \ \ i=1,2,3$ such that 
\begin{equation}
\frac{\eta_i}{\eta_j \eta_k}=\frac{1}{\lambda^2_i}\geq 1, \ \  i\neq j\neq k.
\end{equation}
Now solving the equations for $\eta_i \ \ i=1,2,3$ we find that $\eta_i=\lambda_j \lambda_k$ with $i\neq j\neq k$, this is,
\begin{multline}
\diag(1,\eta_1,\eta_2,\eta_3)=\diag(1,1,\lambda_1,\lambda_1)\diag(1,\lambda_2,1,\lambda_2)\\
\times \diag(1,\lambda_3,\lambda_3,1).
\end{multline}
This means that any infinitesimally divisible Pauli channel is divisible in a bit, phase and bit-phase flip channels, and all of them have Kraus rank $2$. Therefore any infinitesimally divisible non-singular qubit channel is divisible in Kraus rank $2$ channels.
The non-singular case (also with Kraus rank $4$) is proved in the following way. First observe that according to the proof of theorem 5 in Ref.~\cite{davalosdivisibility}, the only infinitesimally divisible Pauli channels (up to unitary conjugations) have the form $\diag(1,0,0,\lambda)$ with $0\leq \lambda \leq 1$. For this case we can use directly the Lindblad-Boundary decomposition and write $\diag(1,0,0,\lambda)=\diag(1,1,\lambda,\lambda)\diag(1,0,0,1)$, both factoring channels have Kraus rank $2$. Therefore any infinitesimally divisible qubit channel is divisible in Kraus rank $2$ channels.
To prove the second part assume now that $\mcE$ is not infinitesimally divisible. According to theorem 4 of our main text, the minimum Kraus rank that a factoring channel can have is $3$, otherwise the channel would be infinitesimally divisible. This finishes the proof.
\end{proof}
\section{Reduction of quantum computer register size to simulate quantum channels}
\label{app:qc1}
Theorem 6 can be used to simulate channels in quantum computers using smaller registers than the ones established by the Stinespring dilation theorem~\cite{Stinespring2006}. In the case of infinitesimally divisible qubit channels, this reduction is enough to dispense with a qubit. This is, full Kraus rank channels need a four-dimensional ancillary system to be simulated (two qubits), but according to theorem $6$ we can do it using only one qubit for infinitesimally divisible channels (since they can be divided in Kraus rank $2$ channels that just need a two-dimensional ancillary system, \ie{} a qubit). Observe that such reduction is impossible for channels that are not infinitesimally divisible since they need at least a three-dimensional ancilla, so we cannot dispense with a qubit. 

\par
As an example consider the channel 
\begin{equation}
\hat \mcE=\diag(1,0,0,\lambda),
\label{eq:channel_implemented}
\end{equation}
with $0<\lambda<1$ which is full Kraus rank. It can be divided in the following way,
\begin{equation}
\hat \mcE=\underbrace{\left(
\begin{array}{cccc}
 1 & 0 & 0 & 0 \\
 0 & 1 & 0 & 0 \\
 0 & 0 & \lambda  & 0 \\
 0 & 0 & 0 & \lambda  \\
\end{array}
\right)}_{e^{t_{\min} \hat \mcL_{\text{bit flip}}}}
\underbrace{\left(
\begin{array}{cccc}
 1 & 0 & 0 & 0 \\
 0 & 0 & 0 & 0 \\
 0 & 0 & 0 & 0 \\
 0 & 0 & 0 & 1 \\
\end{array}
\right)}_{\hat \mcE_\text{boundary}}.
\label{eq:decomp}
\end{equation}
Each factor has Kraus rank $2$ and can be implemented using a two-qubit register, see figure~\ref{fig:circuit}. After implementing the first factoring channel, the ancillary qubit needs to be reset, thus, the circuit needs a bit to classically control the ancillary bit after it is measured. Depending in which state the ancilla collapsed, the bit controls a bit flip on it. After this control, one applies the second unitary corresponding to the Lindblad part.
The unitaries needed for the example of~\eref{eq:channel_implemented}, corresponding to the decomposition of~\eref{eq:decomp}, are the following,
\begin{align}
&U_{boundary}=\left(
\begin{array}{cccc}
 0 & 0 & 1 & 0 \\
 0 & 1 & 0 & 0 \\
 1 & 0 & 0 & 0 \\
 0 & 0 & 0 & 1 \\
\end{array}
\right),\\
&U_{e^\mcL}=\left(
\begin{array}{cccc}
 0 & \frac{\sqrt{1-\lambda }}{\sqrt{2}} & 0 & -\frac{\sqrt{\lambda +1}}{\sqrt{2}} \\
 \frac{\sqrt{1-\lambda }}{\sqrt{2}} & 0 & -\frac{\sqrt{\lambda +1}}{\sqrt{2}} & 0 \\
 \frac{\sqrt{\lambda +1}}{\sqrt{2}} & 0 & \frac{\sqrt{1-\lambda }}{\sqrt{2}} & 0 \\
 0 & \frac{\sqrt{\lambda +1}}{\sqrt{2}} & 0 & \frac{\sqrt{1-\lambda }}{\sqrt{2}} \\
\end{array}
\right),
\end{align}
this is 
$$\mcE_\text{boundary}[\rho]=\tr_\text{ancilla} \left[U_\text{boundary} (\proj{0}{0}\otimes \rho) U^\dagger{}_\text{boundary} \right],$$
and similarly for $e^\mcL$. In the next section we implement this example in a quantum computer.
\begin{figure}
\centering
\includegraphics[width=\columnwidth]{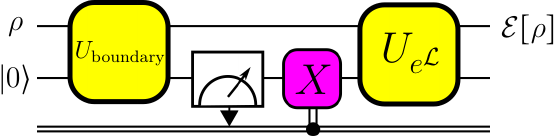}
\caption{Implementation of infinitesimally divisible qubit channels using a two-qubit quantum register. After a measurement in the computational basis is carried on the ancillary qubit, a bit is to prepare the ancillary qubit for the second part of the decomposition. \label{fig:circuit}}
\end{figure}


\begin{widetext}
\section{Implementation in quantum processor}
\label{app:qc2}
In this section we implement the channel given in~\eref{eq:channel_implemented} (with $\lambda=3/4$) in a IBM \textit{falcon r4T} quantum processor~\cite{qiskit}, and provide the \textit{python}~\cite{python} code used to perform quantum process tomography. In particular, the code computes the fidelities between the quantum computed and theoretical Choi matrices, see figure 2 in the main text. 
Unfortunately the falcon r4T processor does not support classically controlled operations, thus, to implement the circuit in~\fref{fig:circuit} we make use of the \textit{deferred} and \textit{implicit} measurement principles~\cite{chuangbook}, see~\fref{fig:circuit_simpler}. This is, we include a third qubit that is used simply as a bit. This is, instead of measuring the ancillary qubit, we use it as a quantum control of the third qubit to ``write'' its state on it, then we use such qubit to control back the ancillary qubit and flip it. Deferred and implicit measurement principles guaranty that the statistics obtained are exactly the same as in circuit of~\fref{fig:circuit}. Since the third qubit is used simply as a bit, the circuit effectively needs only two qubits.
\par
In what follows we show one of the ten experimentally obtained Choi matrices of the channel, and the theoretical one to compare,
\begin{align}
\text{Choi}_\text{experimental}&=\left(
\begin{array}{cccc}
 0.87815 & 0.0104\, -0.03035 i & 0.001725\, -0.040975 i & 0.0119\, -0.011 i \\
 0.0104\, +0.03035 i & 0.12185 & -0.0066-0.0091 i & -0.001725+0.040975 i \\
 0.001725\, +0.040975 i & -0.0066+0.0091 i & 0.1649 & 0.0046\, +0.00495 i \\
 0.0119\, +0.011 i & -0.001725-0.040975 i & 0.0046\, -0.00495 i & 0.8351 \\
\end{array}
\right),\\ 
\text{Choi}_\text{theoretical}&=\left(
\begin{array}{cccc}
 0.875 & 0 & 0 & 0 \\
 0 & 0.125 & 0 & 0 \\
 0 & 0 & 0.125 & 0 \\
 0 & 0 & 0 & 0.875 \\
\end{array}
\right).
\end{align}
\label{sec:a}
\begin{figure}[h!]
\centering
\includegraphics[scale=0.5]{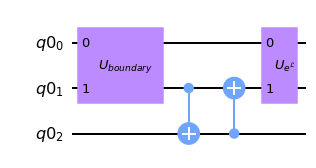}
\caption{Circuit implementing channel of~\eref{eq:channel_implemented} using the Lindblad-Boundary decomposition. The channel is implemented in qubit $q0_0$, while qubit $q0_1$ is used as the ancillary system and $q0_2$ used effectively as a bit. This circuit gives the same statistics in the system and ancillary qubit as of the circuit in~\fref{fig:circuit}. The figure was generated using qiskit. \label{fig:circuit_simpler}}
\end{figure}
{\scriptsize
\begin{lstlisting}[language=Python]e
#Qiskit version: 0.20.1
# Importing libraries for the construction of the quantum circuit
from qiskit import QuantumCircuit, transpile, ClassicalRegister, QuantumRegister, execute
# Importing also numpy, will be needed later
import numpy as np

#Loading account
from qiskit import IBMQ
IBMQ.load_account()
provider = IBMQ.get_provider('ibm-q')

# Tomography function
from qiskit.ignis.verification.tomography import process_tomography_circuits, ProcessTomographyFitter

# job monitor
from qiskit.tools.monitor import job_monitor

#Linear Algebra functions
from numpy import dot, trace
from numpy.linalg import multi_dot
from scipy.linalg import sqrtm

# Defining unitaries needed to implement
U=[[0,0,1,0],
   [0,1,0,0],
   [1,0,0,0],
   [0,0,0,1]
  ]
l=0.75
gm=(1-l)**0.5/(2**0.5)
gp=(1+l)**0.5/(2**0.5)
U2=[[0,gm,0,-gp],
   [gm,0,-gp,0],
   [gp,0,gm,0],
    [0,gp,0,gm]]

# Defining quantum register
q = QuantumRegister(3)
qc=QuantumCircuit(q)
qc.unitary(U,[0,1],'$U_{boundary}$')
qc.cnot([1],[2])
qc.cnot([2],[1])
qc.unitary(U2,[0,1],'$U_{e^{\mathcal{L}}}$')
qc.draw('mpl')

### Quantum process tomography

#Using quantum processor falcon r4T in Lima
backend='ibmq_lima'
quamtum_computer= provider.get_backend(backend)

qpt_circs = process_tomography_circuits(qc,q[0],prepared_qubits=q[0])

# Computing the 12 tomography circuits 20000 each, 
joblist=[]
for i in range(10):
    job = execute(qpt_circs,backend=quamtum_computer, shots=20000)
    print(i)
    job_monitor(job)
    joblist.append(job)

#Getting experimental Choi matrices
chois=list(map(lambda r: ProcessTomographyFitter(r.result(), qpt_circs).fit(),joblist))
#Converting them to numpy array
chois=list(map(np.array,chois))

# Defining the theoretical Choi matrix
theoreticalchoi=np.diag(0.5*np.array([1+l,1-l,1-l,1+l]))

#Computation of fidelities between the theoretical and experimental Choi
fidelities=[]
for k in chois:
    fidelities.append(trace(sqrtm(multi_dot([sqrtm(theoreticalchoi),k,sqrtm(theoreticalchoi)])))/2)

#Print fidelities
fidelities
\end{lstlisting}
}
\end{widetext}
\end{document}